\documentclass[11pt]{article}
\usepackage{graphicx}
\usepackage{longtable,lscape}
\usepackage{amsthm}
\usepackage{amsfonts}
\usepackage{amsmath}
\usepackage{bbm}
\usepackage{float}
\usepackage{url}

\newtheorem{theorem}{Theorem}

\newcommand{\captionfonts}{\footnotesize}
\makeatletter  
\long\def\@makecaption#1#2{%
  \vskip\abovecaptionskip
  \sbox\@tempboxa{{\captionfonts #1: #2}}%
  \ifdim \wd\@tempboxa >\hsize
    {\captionfonts #1: #2\par}
  \else
    \hbox to\hsize{\hfil\box\@tempboxa\hfil}%
  \fi
  \vskip\belowcaptionskip}
\makeatother 
\begin{document}
\title{Identifying Quantum Structures in the Ellsberg Paradox}
\author{Diederik Aerts$^1$, Sandro Sozzo$^1$ and Jocelyn Tapia$^2$ \vspace{0.5 cm} \\ 
        \normalsize\itshape
        $^1$ Center Leo Apostel for Interdisciplinary Studies \\
        \normalsize\itshape
        and, Department of Mathematics, Brussels Free University \\ 
        \normalsize\itshape
         Krijgskundestraat 33, 1160 Brussels, Belgium \\
        \normalsize
        E-Mails: \url{diraerts@vub.ac.be,ssozzo@vub.ac.be}
          \vspace{0.5 cm} \\ 
        \normalsize\itshape
        $^2$ Pontificia Universidad Cat\'olica de Chile \\
        \normalsize\itshape
         Avda. Libertador Bernardo OHiggins 340, Santiago (Chile) \\
        \normalsize
        E-Mails: \url{jntapia@uc.cl} \\
              }
\date{}
\maketitle
\begin{abstract}
\noindent
Empirical evidence has confirmed that quantum effects 
occur frequently also outside the microscopic domain, while quantum structures satisfactorily model various situations in several areas of science, including biological, cognitive and social processes. In this paper, we elaborate a quantum mechanical model which faithfully describes the {\it Ellsberg paradox} in economics, showing that the mathematical formalism of quantum mechanics is capable to represent the {\it ambiguity} present in this kind of situations, because of the presence of {\it contextuality}. Then, we analyze the data collected in a concrete experiment we performed on the Ellsberg paradox and work out a complete representation of them in complex Hilbert space.
We prove that the presence of quantum structure is genuine, that is, 
{\it interference} and {\it superposition} in a complex Hilbert space are really necessary to describe the conceptual situation presented by Ellsberg. Moreover, our approach sheds light on `ambiguity laden' decision processes in economics and decision theory, and allows to deal with different Ellsberg-type generalizations, e.g., the {\it Machina paradox}.
\end{abstract}

\section{Introduction\label{intro}}
Traditional approaches in economics follow the hypothesis that agents' behavior during a decision process is mainly ruled by {\it expected utility theory} (EUT) \cite{vonneumannmorgenstern1944,savage1954}. Roughly speaking, in presence of uncertain events, decision makers choose in such a way that they maximize their utility, or satisfaction. Notwithstanding its mathematical tractability and predictive success, the structural validity of EUT at the individual level is questionable. Indeed, systematic empirical deviations from the predictions of EUT have been observed which are usually referred to as {\it paradoxes} \cite{ellsberg1961,machina2009}.

EUT was formally elaborated by von Neumann and Morgenstern \cite{vonneumannmorgenstern1944}. They presented a set of axioms that allow to represent decision--maker preferences over the set of \emph{acts} (functions from the set of states of the nature into the set of consequences) by a suitable functional $E_p u(.)$, for some 
Bernoulli utility function $u$ on the set of consequences and an objective probability measure $p$ on the set of states of the nature. An important aspect of EUT concerns the treatment of uncertainty. Knight had pointed out the difference between \emph{risk} and \emph{uncertainty} reserving the term \emph{risk} for 
situations that can be described by known (or physical) probabilities, and the term \emph{uncertainty} to refer to situations in which agents do not know the probabilities associated with each of the possible outcomes of an act \cite{knight1921}. Von Neumann and Morgenstern modeling did not contemplate the latter possibility, since all probabilities are {\it objectively}, i.e. physically, given in their scheme. For this reason, Savage extended EUT allowing agents to construct their own subjective probabilities when physical probabilities are not available \cite{savage1954}. According to Savage's model, the distinction proposed by Knight seems however irrelevant. Ellsberg instead showed that Knightian's distinction is empirically meaningful \cite{ellsberg1961}. In particular, he presented the following experiment. Consider one urn with 30 red balls and 60 balls that are either yellow or black, the latter in unknown proportion. One ball will be drawn from the urn. Then, free of charge, a person is asked to bet on one of the acts $f_1$, $f_2$, $f_3$ and $f_4$ defined in Tab. \ref{table01}.
\noindent
\begin{table} \label{table01}
\begin{center} 
\begin{tabular}{| p{1.5cm}|p{1.5cm}|p{1.5cm}|p{1.5cm}|}
\hline
Act & red & yellow & black\\
\hline\hline
$f_1$ & 12\$ & 0\$ & 0\$   \\
\hline
$f_2$ & 0\$ & 0\$ & 12\$  \\
\hline
$f_3$ & 12\$ & 12\$ & 0\$  \\
\hline
$f_4$ & 0\$ & 12\$ & 12\$  \\
\hline
\end{tabular}
\end{center}
\caption{The payoff matrix for the Ellsberg paradox situation.}
\end{table}
\noindent
When asked to rank these gambles most of the persons choose to bet on $f_1$ over $f_2$ and $f_4$ over $f_3$. This preference cannot be explained by EUT. Indeed, individuals' ranking of the sub--acts [12 on {\it red}; 0 on {\it black}] versus [0 on {\it red}; 12 on {\it black}] depends upon whether the event $yellow$ yields a payoff of 0 or 12, contrary to what is suggested by the Sure--Thing principle, an important axiom of Savage's model.\footnote{The Sure--Thing principle was stated by Savage by introducing the {\it businessman example}, but it can be presented in an equivalent form, the {\it independence axiom}, as follows: if persons are indifferent in choosing between simple lotteries $L_1$ and $L_2$, they will also be indifferent in choosing between $L_1$ mixed with an arbitrary simple lottery $L_3$ with probability $p$ and $L_2$ mixed with $L_3$ with the same probability $p$.} Nevertheless, these choices have a direct intuition: $f_1$  offers the 12 prize with an {\it objective probability} of $1/3$, and $f_2$ offers the same prize but in an element of the {\it subjective partition}  $\{black, yellow  \}$. In the same way, $f_4$ offers the prize with an objective probability of $2/3$, whereas $f_3$ offers the same payoff on the union of the unambiguous event {\it red} and the ambiguous event {\it yellow}. Thus, in both cases the unambiguous bet is preferred to its ambiguous counterpart, a phenomenon called {\it ambiguity aversion} by Ellsberg. 

Many extensions of EUT have been worked out to cope with Ellsberg--type preferences, which mainly consist in replacing the Sure--Thing Principle by weaker axioms. We briefly summarize the most known, as follows.

(i) {\it Choquet expected utility}. 
This model considers a subjective non--additive probability (or, capacity) over the states of nature rather than a subjective probability. Thus, decision--makers could underestimate or overestimate probabilities in the Ellsberg experiment, and ambiguity aversion is equivalent to the convexity of the capacity (pessimistic beliefs) \cite{gilboa}.

(ii) {\it Max--Min expected utility}. 
The lack of knowledge about the states of nature of the decision--maker cannot be represented by a unique probability measure but, rather, by a set of probability measures. Then, an act $f$ is preferred to $g$ iff $ \min_{p\in P} E_p u(f) >  \min_{p\in P} E_p u(f)$, where $P$ is a convex and closed set of additive probability measures. Ambiguity aversion is represented by the pessimistic beliefs of the agent which takes decisions considering the worst probabilistic scenario \cite{gilboaschmeidler1989}. 

(iii) {\it Variational preferences}. In this dynamic generalization of the Max--Min expected utility, agents rank acts according to the criterion $\inf_{p\in \bigtriangleup } \{E_p u(f)+c(p)\}$, where $c(p)$ is a closed and convex penalty function associated with the probability election \cite{mmr2006}.
 
(iv) {\it Second order probabilities}. This is a model of preferences over acts where the decision--maker prefers act $f$ to act $g$ iff $E_{\mu} \phi (E_p u (f) ) > E_{\mu} \phi$ $ (E_p u (g))$, where $E$ is the expectation operator, $u$ is a von Neumann--Morgenstern utility function, $\phi$ is an increasing transformation, and $\mu$ is a subjective probability over the set of probability measures $p$ that the decision--maker thinks are feasible. Ambiguity aversion is here represented by the concavity of the transformation $\phi$ \cite{kmm2005}. 

Despite approaches (i)--(iv) have been widely used in economic and financial modeling, none of them is immune of critics \cite{machina2009,epstein1999}. And, worse, none of these models can satisfactorily represent more general Ellsberg--type situations (e.g., the {\it Machina paradox} \cite{machina2009,bdhp2011}). As a consequence, the construction of a unified perspective representing ambiguity is still an unachieved goal in economics and decision making.

We have recently inquired both conceptually and structurally into the above approaches generalizing EUT \cite{gilboa,gilboaschmeidler1989,mmr2006,kmm2005} to cope with ambiguity. The latter is defined as a situation without a probability model describing it as opposed to \emph{risk}, where a classical probability model on a $\sigma$--algebra of events is presupposed. The generalizations in (i)--(iv) consider more general structures than a single classical probability model on a $\sigma$--algebra. We are convinced that this is the point: ambiguity, due to its contextuality, structurally need a non--classical probability model. To this end we have elaborated a general framework, based on the notion of {\it contextual risk} and inspired by the probability structure of quantum mechanics, which is intrinsically different from a classical probability on a $\sigma$-algebra, the set of events is indeed {\it not} a Boolean algebra \cite{aertsczachordhooghe2011,aertsdhooghesozzo2011,aertssozzovaxjo,aertssozzokavala12}. Inspired by this approach, we work out in the present article a complete mathematical representation of the Ellsberg paradox situation (states, payoffs, acts, preferences) in the standard mathematical formalism of quantum mechanics, hence by using a complex Hilbert space, and representing the probability measures by means of projection valued measures on this complex Hilbert space \cite{ast12} (Sect. \ref{ellsberg}). 
This analysis leads us to claim that the structure of the probability models is essentially different from the ones of known approaches -- projection valued measures instead of $\sigma$--algebra valued measures \cite{gudder}. But, more important, also the way in which states are represented in quantum mechanics, i.e. by unit vectors of the Hilbert space, introduces a fundamentally different aspect, coping both mathematically and intuitively with the notion of ambiguity as introduced in economics. Successively, we analyze the experimental data we collected in a {\it statistically
relevant} experiment we performed, where real decision--makers were asked to bet on the different acts defined by the Ellsberg paradox situation \cite{aertsczachordhooghe2011,aertsdhooghesozzo2011} (Sect. \ref{experiment}). We 
show that our quantum mechanical model faithfully represents the subjects' preferences together with experimental statistics. Furthermore, we describe the choices between acts $f_1$/$f_2$ and between acts $f_3$/$f_4$ by quantum observables represented by compatible spectral families (Sect. \ref{mod_experiment}). We finally 
prove that the requirement of compatibility of the latter observables makes it necessary to introduce a Hilbert space over 
{\it complex} numbers, namely that {\it imaginary} numbers are needed since our experimental data 
cannot be modeled in a real Hilbert space, 
i.e. a vector space over {\it real} numbers only, in case we want the observables representing our experiment to be compatible (Sect. \ref{realcase}). Complex numbers in quantum theory stand for the quantum effect of interference, and indeed, we can identify in our modeling how interference produces the typical Ellsberg deviation leading to measured data in our experiment. 
These results strongly suggest that, more generally, `ambiguity laden' situations can be explained in terms of the appearance of typically quantum effects, namely, contextuality -- our quantum model is contextual, see the discussion on context in Sect. \ref{ellsberg} --, superposition -- we explicitly use superposition to construct the quantum states representing the Ellsberg bet situations in Sect. \ref{ellsberg} -- and interference -- i.e. complex numbers --, and that quantum structures have the capacity to mathematically deal with this type of situations. Hence, our findings naturally fit within the growing `quantum interaction research' which mainly applies quantum structures to cognitive situations \cite{bruzaetal2007,bruzaetal2008,bruzaetal2009,aerts2009,busemeyerlambert2009,pb2009,danilovlambert2010,k2010,songetal2011,bpft2011,bb2012}. 

\section{A quantum model in Hilbert space for the Ellsberg paradox\label{ellsberg}}
We have recently worked out a Hilbert space model for the Ellsberg paradox situation \cite{ast12}. Here we deepen our Hilbert space representation of the Ellsberg situation and derive 
new results elaborating the already obtained ones.
But we first need to anticipate a discussion on the notions of `context' and `contextual influence' and how they are used in the present framework. The notion of context typically denotes what does not pertain to the entity under study but that can interact with it. In the foundations of quantum mechanics, context more specifically indicates the `measurement context', which influences the measured quantum entity in a stochastic way. As a consequence of this contextual interaction, the state of the quantum entity changes, thus determining a transition from potential to actual. A quantum mechanical context is represented by a self-adjoint operator or, equivalently, by a spectral family. An analogous effect occurs in a decision process, where there is generally a contextual influence (of a cognitive nature) having its origin in the way the mind of the person involved in the decision, e.g., a choice between two bets, relates to the situation that is the subject of the decision making, e.g., the Ellsberg situation. This is why, in our analysis of the Ellsberg paradox, we use the definition and representation of context and contextual influence to indicate the cognitive interaction taking place between the conceptual Ellsberg situation and the human mind in a decision process. We are now ready to proceed with our quantum modeling.

To this end let us consider the situation illustrated in Tab. 1, Sect. \ref{intro}. The simplest Hilbert space that can supply a faithful modeling of the Ellsberg situation is the three dimensional complex Hilbert space $\mathbb{C}^3$ whose canonical basis we denote by $\{|1,0,0\rangle, |0,1,0\rangle, |0,0,1\rangle \}$. We introduce the model in different steps: we will see that in each of them, quantum structures enable a new and satisfying way to model an aspect of ambiguity. At the last step it will be made clear that modeling the agents' decisions at a statistical level requires the full probabilistic apparatus of quantum mechanics, that is, both states and measurements.

\subsection{The conceptual Ellsberg entity\label{step1}}
The first part of the model consists of the Ellsberg situation without considering neither the different acts nor the person nor the bet to be taken. Hence it is the situation of the urn with 30 red balls and 60 black and yellow balls in unknown proportion ({\it conceptual Ellsberg entity}). Already at this stage, the presence of ambiguity can be mathematically taken into account by means of the quantum mechanical formalism. To this aim we introduce a quantum mechanical context $e$ and represent it by means of the family $\{P_r, P_{yb}\}$, where $P_r$ is the one dimensional orthogonal projection operator on the subspace generated by the unit vector $|1,0,0\rangle$, and $P_{yb}$ is the two dimensional orthogonal projection operator on the subspace generated by the unit vectors $|0,1,0\rangle$ and $|0,0,1\rangle$. $\{P_r, P_{yb}\}$ is a spectral family, since $P_r \perp P_{yb}$ and $P_r+P_{yb}=\mathbbmss{1}$. Contexts, more specifically measurement contexts, are represented by spectral families of orthogonal projection operators (equivalently, by a self--adjoint operator determined by such a family) also in quantum mechanics. Again in analogy with quantum mechanics, we represent the states of the conceptual Ellsberg entity by means of unit vectors of ${\mathbb C}^{3}$. For example, the unit vector
\begin{equation} \label{stateredyellow}
|v_{ry}\rangle=|{1 \over \sqrt{3}}e^{i\theta_r}, \sqrt{2 \over 3}e^{i\theta_y}, 0\rangle
\end{equation}
can be used to represent a state describing the Ellsberg situation. Indeed, the probability for `red' in the state $p_{v_{ry}}$ represented by $|v_{ry}\rangle$ is
\begin{eqnarray}
|\langle 1,0,0|v_{ry}\rangle|^2=\langle v_{ry}|1,0,0\rangle\langle 1,0,0|v_{ry}\rangle=\parallel P_{r}|v_{ry}\rangle\parallel^{2}={1 \over 3}
\end{eqnarray}
Moreover, the probability for `yellow or black' in the state $p_{v_{ry}}$ represented by $|v_{ry}\rangle$ is
\begin{equation}
\parallel P_{yb}|v_{ry}\rangle \parallel^{2}=\langle 0,\sqrt{2 \over 3}e^{i\theta_y},0|0,\sqrt{2 \over 3}e^{i\theta_y},0\rangle={2 \over 3}
\end{equation}
But this is not the only state describing the Ellsberg situation. For example, the unit vector
\begin{equation} \label{stateredblack}
|v_{rb}\rangle=|{1 \over \sqrt{3}}e^{i\phi_r}, 0, \sqrt{2 \over 3}e^{i\phi_b}\rangle
\end{equation}
also represents a state describing the Ellsberg situation. We thus denote the set of all states describing the Ellsberg situation (\emph{Ellsberg state set}) by
\begin{equation}
\Sigma_{Ells}=\{p_{v}:  \ |v\rangle=  |{1 \over \sqrt{3}}e^{i\theta_r},\rho_ye^{i\theta_y},\rho_be^{i\theta_b}\rangle \ \vert \ 0\le\rho_y,\rho_b,\ \rho_y^2+\rho_b^2={2 \over 3}\}
\end{equation}
which is associated with a subset (not necessarily a linear subspace) of $\mathbb{C}^3$. If a state belongs to $\Sigma_{Ells}$, this state delivers a quantum description of the Ellsberg situation, together with the context $e$ represented by the spectral family $\{P_r, P_{yb}\}$ in $\mathbb{C}^3$.

\subsection{Modeling acts and utility\label{step2}}
In the second step of our construction, we describe the different acts $f_1$, $f_2$, $f_3$ and $f_4$. Here a second measurement context $g$ is introduced. The context $g$ describes the ball taken out of the urn and its color verified, red, yellow or black. Also $g$ is represented by a spectral family of orthogonal projection operators $\{P_r, P_y, P_b\}$, where $P_r$ is already defined, while $P_y$ is the orthogonal projection operator on $|0,1,0\rangle$ and $P_b$ is the orthogonal projection operator on $|0, 0, 1\rangle$. Thus, the probabilities $\mu_r(g,p_{v})$, $\mu_y(g,p_{v})$ and $\mu_b(g,p_{v})$ of drawing a red ball, a yellow ball and a black ball, respectively, in a state $p_{v}$ represented by the unit vector $|v\rangle=|\rho_re^{i\theta_r},\rho_ye^{i\theta_y},\rho_be^{i\theta_b}\rangle$ are
\begin{eqnarray}
\mu_r(g,p_{v})=\parallel P_r|v\rangle\parallel^{2}=\langle v|P_{r}|v\rangle=\rho_r^2 \\
\mu_y(g,p_{v})=\parallel P_y| v\rangle\parallel^{2}=\langle v|P_{y}|v\rangle=\rho_y^2 \\
\mu_b(g,p_{v})=\parallel P_b|v\rangle\parallel^{2}=\langle v|P_{b}|v\rangle =\rho_b^2
\end{eqnarray}

The acts $f_1$, $f_2$, $f_3$ and $f_4$ are observables in our modeling, hence they are represented by self-adjoint operators, built on the spectral decomposition $\{P_r, P_y, P_b\}$. More specifically, we have
\begin{eqnarray}
{\mathcal F}_1&=&12\$P_r \\
{\mathcal F}_2&=&12\$P_b \\
{\mathcal F}_3&=&12\$P_r+12\$P_y \\
{\mathcal F}_4&=&12\$P_y+12\$P_b=12\$P_{yb}
\end{eqnarray}
Let us now analyze the expected payoffs and the utility connected with the different acts. For the sake of simplicity, we identify here the utility with the expected payoff, which implies that we are considering {\it risk neutral agents}. 
Consider an arbitrary state $p_{v} \in \Sigma_{Ells}$ and the acts $f_1$ and $f_4$. We have
\begin{eqnarray}
U(f_1,g,p_{v})&=&\langle p_{v}|{\mathcal F}_1|p_{v} \rangle
=12\$\cdot {1 \over 3}=4\$ \\
U(f_4,g,p_{v})&=&\langle p_{v}|{\mathcal F}_4|p_{v} \rangle
=12\$\cdot {2 \over 3}=8\$
\end{eqnarray} 
which shows that both these utilities are completely {\it independent} of the considered state of $\Sigma_{Ells}$, i.e. they are {\it ambiguity free}. Consider now the acts $f_2$ and $f_3$, and again an arbitrary state $p_{v} \in \Sigma_{Ells}$. We have
\begin{eqnarray}
U(f_2,g,p_{v})&=&\langle p_{v}|{\mathcal F}_2|p_{v} \rangle
=12\$\mu_b(g,p_{v}) \\
U(f_3,g,p_{v})&=&\langle p_{v}|{\mathcal F}_3|p_{v} \rangle
=12\$(\mu_r(g,p_{v})+\mu_y(g,p_{v}))
\end{eqnarray}
which shows that both utilities strongly depend on the state $p_{v}$, due to the ambiguity the two acts are confronted with. This can be significantly revealed by considering two extreme cases. Let $p_{v_{ry}}$ and $p_{v_{rb}}$ be the states represented by the vectors $|v_{ry}\rangle$ and $|v_{rb}\rangle$ in Eqs. (\ref{stateredyellow}) and (\ref{stateredblack}), respectively. These states give rise for the act $f_2$ to utilities
\begin{eqnarray}
U(f_2,g,p_{v_{ry}})&=&12\$\mu_b(g,p_{v_{ry}})=12\$\cdot0=0\$ \\
U(f_2,g,p_{v_{rb}})&=&12\$\mu_b(g,p_{v_{rb}})=12\$\cdot {2 \over 3}=8\$.
\end{eqnarray}
This shows that a state $p_{v_{rb}}$ exists within the realm of ambiguity, where the utility of act $f_2$ is greater than the utility of act $f_1$, and also a state $p_{v_{ry}}$ exists within the realm of ambiguity, where the utility of act $f_2$ is smaller than the utility of act $f_1$. If we look at act $f_3$, we find for the two considered extreme states the following utilities
\begin{eqnarray}
U(f_3,g,p_{v_{ry}})&=&12\$(\mu_r(g,p_{v_{ry}})+\mu_y(g,p_{v_{ry}}))=12\$({1 \over 3}+{2 \over 3})=12\$ \\
U(f_3,g,p_{v_{rb}})&=&12\$(\mu_r(g,p_{v_{rb}})+\mu_y(g,p_{v_{rb}}))=12\$({1 \over 3}+0)=4\$.
\end{eqnarray}
Analogously, namely the state $p_{v_{ry}}$ gives rise to a greater utility, while the state $p_{v_{rb}}$ gives rise to a smaller utility than the independent one obtained in act $f_4$. 

\subsection{Decision making and superposition states\label{step3}}
The third step of our modeling consists in taking directly into account the role played by ambiguity. We proceed as follows \cite{aerts2009}. 

We suppose that the two extreme states $p_{v_{ry}}$ and $p_{v_{rb}}$ in Sect. \ref{step1} are relevant in the mind of the person that is asked to bet. Hence, it is a superposition state of these two states that will guide the decision process during the bet. Let us construct a general superposition state $p_{v_s}$ of these two states. Hence the vector $|v_s\rangle$ representing $p_{v_s}$ can be written as 
\begin{equation}\label{superpositionEllsberg}
|v_s\rangle=ae^{i\alpha}|v_{rb}\rangle+be^{i\beta}|v_{ry}\rangle
\end{equation}
where $a$, $b$, $\alpha$ and $\beta$ are chosen in such a way that $\langle v_s|v_s\rangle=1$, which means that
\begin{equation}
1=a^2+b^2+{2ab \over 3}\cos(\beta-\alpha+\theta_r-\phi_r)
\end{equation}
or, equivalently,
\begin{equation}
\cos(\beta-\alpha+\theta_r-\phi_r)={3(1-a^2-b^2) \over 2ab}
\end{equation}
The amplitude of the state $p_{v_s}$ with the first basis vector $|1,0,0\rangle$ is given by
\begin{equation}
\langle 1,0,0|v_s\rangle={a \over \sqrt{3}}e^{i(\alpha+\phi_r)}+{b \over \sqrt{3}}e^{i(\beta+\theta_r)}
\end{equation}
 Therefore, the transition probability is
\begin{equation}
|\langle 1,0,0|v_s\rangle|^2={1 \over 3}(a^2+b^2+3-3a^2-3b^2)={1 \over 3}(3-2a^2-2b^2)=\mu_r(g,p_{v_s})
\end{equation}
as one can easily verify. Analogously, the amplitudes with the second and third basis vectors are given by
\begin{eqnarray}
\langle 0,1,0|v_s\rangle&=&ae^{i\alpha}\langle 0,1,0|v_{rb}\rangle+be^{i\beta}\langle 0,1,0|v_{ry}\rangle=\sqrt{{2 \over 3}}be^{i(\beta+\theta_y)} \\
\langle 0,0,1|v_s\rangle&=&ae^{i\alpha}\langle 0,0,1|v_{rb}\rangle+be^{i\beta}\langle 0,0,1|v_{ry}\rangle=\sqrt{{2 \over 3}}ae^{i(\alpha+\theta_b)}
\end{eqnarray}
respectively.
Therefore, the transition probabilities are
\begin{eqnarray}
|\langle 0,1,0|v_s\rangle|^2&=&{2 \over 3}b^2=\mu_y(g,p_{v_s}) \\
|\langle 0,0,1|v_s\rangle|^2&=&{2 \over 3}a^2=\mu_b(g,p_{v_s})
\end{eqnarray}
From the foregoing follows that a general superposition state is represented by the unit vector
\begin{equation}
|v_s\rangle={1 \over \sqrt{3}}|ae^{i(\alpha+\phi_r)}+be^{i(\beta+\theta_r)}, \sqrt{2}be^{i(\beta+\theta_y)}, \sqrt{2}ae^{i(\alpha+\theta_b)}\rangle
\end{equation}
and that the utilities corresponding to 
the different acts are given by
\begin{eqnarray}
U(f_1,g,p_{v_s})&=&\langle v_s|{\mathcal F}_1|v_s\rangle=
4\$(3-2a^2-2b^2) \\
U(f_2,g,p_{v_s})&=&\langle v_s|{\mathcal F}_2|v_s\rangle=
4\$\cdot2a^2 \\
U(f_3,g,p_{v_s})&=&\langle v_s|{\mathcal F}_3|v_s\rangle=
4\$\cdot(3-2a^2) \\
U(f_4,g,p_{v_s})&=&\langle v_s|{\mathcal F}_4|v_s\rangle=
4\$(2a^2+2b^2) 
\end{eqnarray}
We can see that it is not necessarily the case that $\mu_r(g,p_{v_s})={1 \over 3}$, which means that choices of $a$ and $b$ can be made such that the superposition state $p_{v_s}\notin \Sigma_{Ells}$. The reason is that $\Sigma_{Ells}$ is not a linearly closed subset of $\mathbb{C}^3$. 

Let us then consider some of the extreme possibilities of superpositions. 
We remind that the latter superpositions are not relevant for the modeling of the Ellsberg paradox as it was originally formulated, but they come into play if one wants to represent more general Ellsberg-type situations. For example, choose $a=b={\sqrt{3} \over 2}$. Then we have $\mu_y(g,p_{v_s})={2 \over 3}\cdot{3 \over 4}={1 \over 2}$, $\mu_b(g,p_{v_s})={2 \over 3}\cdot{3 \over 4}={1 \over 2}$, and $\mu_r(g,p_{v_s})=0$, and $\cos(\beta-\alpha+\theta_r-\phi_r)={3(1-a^2-b^2) \over 2ab}=-1$, hence $\beta=\pi+\alpha-\theta_r+\phi_r$. Thus, the state represented by
\begin{equation}
|v_s(yb)\rangle={\sqrt{3} \over 2}(e^{i\alpha}|v_{rb}\rangle+e^{i(\pi+\alpha-\theta_r+\phi_r)}|v_{ry}\rangle)
\end{equation}
gives rise to probability zero for a red ball to be drawn. Another extreme choice is, when we take $a=b=\sqrt{3 \over 8}$. Then we have $\mu_y(g,p_{v_s})={2 \over 3}\cdot{3 \over 8}={1 \over 4}$, $\mu_b(g,p_{v_s})={2 \over 3}\cdot{3 \over 8}={1 \over 4}$ and $\mu_r(g,p_{v_s})={1 \over 2}$, and $\cos(\beta-\alpha+\theta_r-\phi_r)=+1$, which means $\beta=\alpha-\theta_r+\phi_r$. Hence, the state
\begin{equation}
|v_s(r)\rangle=\sqrt{3 \over 8}(e^{i\alpha}|v_{rb}\rangle+e^{i(\alpha-\theta_r+\phi_r)}|v_{ry}\rangle)
\end{equation}
gives rise to probability ${1 \over 2}$ for a red ball to be drawn. These are extreme superposition states which are not compatible with the situation as formulated by Ellsberg, but they could be useful if suitable Ellsberg--type extensions are taken into account. 

To construct non--trivial superpositions that retain probability ${1 \over 3}$ for drawing a red ball, we require that
\begin{equation}
{1 \over 3}=\mu_r(g,p_{v_s})={1 \over 3}(3-2a^2-2b^2)
\end{equation}
or, equivalently,
\begin{equation}
a^2+b^2=1
\end{equation}
which implies that
$\cos(\beta-\alpha+\theta_r-\phi_r)=0$,
hence
$\beta={\pi \over 2}+\alpha-\theta_r+\phi_r$.

Let us construct now two examples of superposition states that conserve the ${1 \over 3}$ probability for drawing a red ball, and hence are conservative superpositions, and express ambiguity as is thought to be the case in the Ellsberg paradox situation. The first state refers to the comparison for a bet between $f_1$ and $f_2$. The ambiguity of not knowing the number of yellow and black balls in the urn, only their sum to be 60, as compared to knowing the number of red balls in the urn to be 30, gives rise to the thought that `eventually there are perhaps almost no black balls and hence an abundance of yellow balls'. Jointly, and in superposition, the thought also comes that `it is of course also possible that there are more black balls than yellow balls'. These two thoughts in superposition, are mathematically represented by a state $p_{v_s}$. The state $p_{v_s}$ will be closer to $p_{v_{ry}}$, the extreme state with no black balls, if the person is deeply ambiguity averse, while it will be closer to $p_{v_{rb}}$, the extreme state with no yellow balls, if the person is attracted by ambiguity. Hence, these two tendencies are expressed by the values of $a$ and $b$ in the superposition state. If we consider again the utilities, this time with $a^2+b^2=1$, we have
\begin{eqnarray}
U(f_1,g,p_{v_s})&=&4\$ \\
U(f_2,g,p_{v_s})&=&4\$\cdot2a^2 \\
U(f_3,g,p_{v_s})&=&4\$\cdot(3-2a^2) \\
U(f_4,g,p_{v_s})&=&8\$
\end{eqnarray}
So, for $a^2 < {1 \over 2}$, which exactly means that the superposition state $p_{v_s}$ is closer to the state $p_{v_{ry}}$ than to the state $p_{v_{rb}}$, we have that $U(f_2,g,p_{v_s}) < U(f_1,g,p_{v_s})$, and hence a person with strong ambiguity aversion in the situation of the first bet, will then prefer to bet on $f_1$ and not on $f_2$. Let us choose a concrete state for the bet between $f_1$ and $f_2$, and call it $p_{v_s^{12}}$, and denote its superposition state by $|v_s^{12}\rangle$. Hence, for $|v_s^{12}\rangle$ we take $a={1 \over 2}$ and $b={\sqrt{3} \over 2}$ and hence $a^2={1 \over 4}$ and $b^2={3 \over 4}$. For the angles we must have $\beta-\alpha+\theta_r-\phi_r={\pi \over 2}$, hence let us choose $\theta_r=\phi_r=0$, $\alpha=0$, and $\beta={\pi \over 2}$. This gives us
\begin{equation} \label{vs12}
|v_s^{12}\rangle={1 \over 2\sqrt{3}}|1+\sqrt{3}e^{i{\pi \over 2}}, \sqrt{2}\sqrt{3}e^{i{\pi \over 2}}, \sqrt{2}\rangle={1 \over 2\sqrt{3}}|1+i\sqrt{3}, i\sqrt{6}, \sqrt{2}\rangle
\end{equation}
On the other hand, for ${1 \over 2} < a^2$, which means that the superposition state is closer to the state $p_{v_{rb}}$ than to the state $p_{v_{ry}}$, we have that $U(f_3,g,p_{v_s}) < U(f_4,g,p_{v_s})$, and hence a person with strong ambiguity aversion in the situation of the second bet, will then prefer to bet on $f_4$ and not on $f_3$. Also for this case we construct an explicit state, let us call it $p_{v_s^{23}}$, and denote it by the vector $|v_s^{34}\rangle$. Hence, for $|v_s^{34}\rangle$ we take $a={\sqrt{3} \over 2}$ and $b={1 \over 2}$ and hence $a^2={3 \over 4}$ and $b^2={1 \over 4}$. For the angles we must have $\beta-\alpha+\theta_r-\phi_r={\pi \over 2}$, hence let us choose $\theta_r=\phi_r=0$, $\alpha=0$, and $\beta={\pi \over 2}$. This gives us
\begin{equation} \label{vs34}
|v_s^{34}\rangle={1 \over 2\sqrt{3}}|\sqrt{3}+e^{i{\pi \over 2}}, \sqrt{2}e^{i{\pi \over 2}}, \sqrt{2}\sqrt{3}\rangle={1 \over 2\sqrt{3}}|\sqrt{3}+i, i\sqrt{2}, \sqrt{6}\rangle
\end{equation}
The superposition states $p_{v_{s}}^{12}$ and $p_{v_{s}}^{34}$ representing the unit vectors $|v_s^{12}\rangle$ and $|v_s^{34}\rangle$, respectively, will be used in the next sections to provide a faithful modeling of a concrete experiment in which decisions are expressed by real agents.

\section{An experiment testing the Ellsberg paradox\label{experiment}} 	
We have observed in the previous sections that genuine quantum aspects intervene in the description of the Ellsberg paradox. This will be even more evident from the analysis of a {\it statistically relevant} experiment we performed of this paradox, whose results were firstly reported in \cite{aertsdhooghesozzo2011}. To perform the experiment we sent out the following text to several people, consisting of a mixture of friends, relatives and students, to avoid as much as possible a statistical 
selection bias. 
	
\emph{We are conducting a small-scale statistics investigation into a particular problem and would like to invite you to participate as test subjects. Please note that it is not the aim for this problem to be resolved	in terms of correct or incorrect answers. It is your preference for a particular choice we want to test. The	
 question concerns the following situation.	Imagine an urn containing 90 balls of three different colors: red balls, black balls and yellow balls. We	  know that the number of red balls is 30 and that the sum of the the black balls and the yellow balls is 60.	 The questions of our investigation are about the situation where somebody randomly takes one ball from the urn.}	

(i) \emph{The first question is about a choice to be made between two bets: bet $f_1$ and bet $f_2$. Bet $f_1$ involves winning `10 euros when the ball is red' and `zero euros when it is black or yellow'. Bet $f_2$ involves winning `10 euros when the ball is black' and `zero euros when it is red or yellow'. The 
question we would ask	you to answer is: Which of the two bets, bet $f_1$ or bet $f_2$, would you prefer?}
	
(ii) \emph{The second question is again about a choice between two different bets, bet $f_3$ and bet $f_4$. Bet $f_3$	 involves winning `10 euros when the ball is red or yellow' and `zero euros when the ball is black'. Bet $f_4$	involves winning `10 euros when the ball is black or yellow' and `zero euros when the ball is red'. The	
 second question therefore is: Which of the two bets, bet $f_3$ or bet $f_4$, would you prefer?}
	
\emph{Please provide in your reply message the following information.}

\emph{For question 1, your preference (your choice between bet $f_1$ and bet $f_2$). For question 2, your preference (your choice between bet $f_3$ and bet $f_4$).	
By `preference' we mean `the bet you would take if this situation happened to you in real life'. You are	
expected to choose one of the bets for each of the questions, i.e. `not choosing is no option'. You are welcome to provide a brief explanation of your preferences, which may be of a purely intuitive	nature, only mentioning feelings, for example, but this is not required. It is all right if you only state your preferences without giving any explanation.}
	
\emph{One final remark about the colors. Your choices should not be affected by any personal color preference. If you feel that the colors of the example somehow have an 
influence on your choices, you should restate the	problem and take colors that are indifferent to yours, if this does not work, 
use other neutral characteristics	
to distinguish the balls.}	

Let us now analyze the obtained results.
	
We had 59 respondents participating in our test of the Ellsberg paradox problem, which is the typical number of participants in experiments on psychological effect of the type studied by 
Kahneman and Tversky, such as the 
conjunction fallacy, and the disjunction effect. (see, e.g., \cite{tversky1982,tversky1992}). We do believe that ambiguity aversion is a psychological effect within this calls of effects, which means that in case our hypothesis on the nature of ambiguity aversion is correct, our test is significant. This being said, it would certainly be interesting to make a similar test with a larger number of participants, which is something we plan for the future. We however also remark that the quantum modeling scheme is general enough to also model statistical data that are different from the ones collected in this specific test. Next to this remark, we want to point out that in the present paper we want to prove that these real data `can' be modeled in our approach. 

The answers of the participants were distributed as follows:  (a) 34 subjects preferred	bets $f_1$ and $f_4$; (b) 12 subjects preferred bets $f_2$ and $f_3$; (c) 7 subjects preferred bets $f_2$ and $f_4$; (d) 6 subjects preferred bets $f_1$ and $f_3$. This makes the weights with preference of bet $f_1$ over bet $f_2$ to be 0.68 against 0.32, and the weights with preference of bet $f_4$ over bet $f_3$ to be 0.69 against 0.31. It is worth to note that 34+12=46 people	chose the combination of bet $f_1$ and bet $f_4$ or bet  $f_2$ and bet $f_3$, which is 
78\%. In Sect. \ref{mod_experiment} we apply our quantum mechanical model to these experimental data. 

\section{Quantum modeling the experiment\label{mod_experiment}}
As anticipated in Sect. \ref{step3}, we take into account the superposition states $p_{v_{s}}^{12}$ and $p_{v_{s}}^{34}$ to put forward a description of the choices made by the participants in the test described in Sect. \ref{experiment}. We employ in this section the spectral methods that are typically used in quantum mechanics to construct self--adjoint operators.

First we consider the choice to bet on $f_1$ or on $f_2$. This is a choice with two possible outcomes, let us call them $o_1$ and $o_2$. Thus, we introduce two projection operators $P_1$ and $P_2$ on the Hilbert space ${\mathbb C}^{3}$, and represent the observable associated with the first bet by the self--adjoint operator (spectral decomposition) ${\mathcal O}_{12}=o_1 P_1+o_2P_2$. Then, we consider the choice to bet on $f_3$ or on $f_4$. This is a choice with two possible outcomes too, let us call them $o_3$ and $o_4$. Thus, we introduce two projection operators $P_3$ and $P_4$ on ${\mathbb C}^{3}$, and represent the observable associated with the second bet by the self--adjoint operator (spectral decomposition) ${\mathcal O}_{34}=o_3 P_3+o_4P_4$.

To model the empirical data in Sect. \ref{experiment}, we recall that we tested in our experiment 59 participants, and 40 preferred $f_1$ over $f_2$, while the remaining 19 preferred $f_2$ over $f_1$. This means that we should construct $P_1$ and $P_2$ in such a way that
\begin{equation}
\langle v_s^{12}|P_1|v_s^{12}\rangle={40 \over 59}=0.68 \quad  \langle v_s^{12}|P_2|v_s^{12}\rangle={19 \over 59}=0.32
\end{equation} 
In our experiment of the 59 participants there were 41 preferring $f_4$ over $f_3$, and 18 who choose the other way around. Hence, we should have
\begin{equation}
\langle v_s^{34}|P_3|v_s^{34}\rangle={18 \over 59}=0.31 \quad  \langle v_s^{34}|P_4|v_s^{34}\rangle={41 \over 59}=0.69
\end{equation} 
Both bets should give rise to no preference, hence probabilities ${1 \over 2}$ in all cases, when there is no ambiguity, when the state is $p_{v_c}$ represented by the vector
\begin{equation} \label{nobias}
|v_c\rangle= | {1 \over \sqrt{3}},{1 \over \sqrt{3}},{1 \over \sqrt{3}} \rangle
\end{equation}
Let us preliminarily denote by ${\mathcal S}_{1}^{12}$, ${\mathcal S}_{2}^{12}$,  ${\mathcal S}_{3}^{34}$ and ${\mathcal S}_{4}^{34}$ the eigenspaces associated with the eigenvalues $o_1$, $o_2$, $o_3$ and $o_4$, respectively. We can assume that ${\mathcal S}_{2}^{12}$ and ${\mathcal S}_{4}^{34}$ are one dimensional, and ${\mathcal S}_{1}^{12}$ and ${\mathcal S}_{3}^{34}$ are two dimensional, without loss of generality. Then, we prove the following two theorems.
\begin{theorem} \label{th1}
If, under the hypothesis on the eigenspaces formulated above, the two self--adjoint operators (spectral decompositions) ${\mathcal O}_{12}=o_1 P_1+o_2P_2$ and ${\mathcal O}_{34}=o_3 P_3+o_4P_4$ are such that $[{\mathcal O}_{12}, {\mathcal O}_{34}]=0$ in ${\mathbb C}^{3}$, then
two situations are possible, (i) an orthonormal basis $\{ |e_1\rangle, |e_2\rangle, |e_3\rangle \}$ of common eigenvectors exists such that $P_1=|e_1\rangle\langle e_1|+|e_2\rangle\langle e_2|$, $P_2=|e_3\rangle\langle e_3|$, $P_3=|e_2\rangle\langle e_2|+|e_3\rangle\langle e_3|$ and $P_4=|e_1\rangle\langle e_1|$, or (ii) ${\mathcal O}_{12}={\mathcal O}_{34}$, and hence ${\mathcal S}_{2}^{12}={\mathcal S}_{4}^{34}$ and ${\mathcal S}_{1}^{12}={\mathcal S}_{3}^{34}$.
\end{theorem}
\begin{proof}
Suppose that ${\mathcal O}_{12}\not={\mathcal O}_{34}$. Since ${\mathcal S}_{2}^{12}$ and ${\mathcal S}_{4}^{34}$ are one dimensional, we can choose $|e_1\rangle$ and $|e_3\rangle$ unit vectors respectively in ${\mathcal S}_{2}^{12}$ and ${\mathcal S}_{4}^{34}$. Since $[{\mathcal O}_{12}, {\mathcal O}_{34}]=0$ it follows that $[P_2, P_4]=0$, and hence $|e_1\rangle \perp |e_3\rangle$ or $|e_1\rangle=\lambda|e_2\rangle$ for some $\lambda \in {\mathbb C}$. But, if $|e_1\rangle=\lambda|e_2\rangle$, we have ${\mathcal S}_{2}^{12}={\mathcal S}_{4}^{34}$, and hence ${\mathcal S}_{1}^{12}=({\mathcal S}_{2}^{12})^\perp=({\mathcal S}_{4}^{34})^\perp={\mathcal S}_{3}^{34}$, which entails ${\mathcal O}_{12}={\mathcal O}_{34}$. Hence, we have $|e_1\rangle \perp |e_3\rangle$. Since ${\mathcal S}_{1}^{12}$ and ${\mathcal S}_{3}^{34}$ are both two dimensional, their intersection is one dimensional, or they are equal. If they are equal, then also their orthogocomplements are equal, and this would entail again that ${\mathcal O}_{12}={\mathcal O}_{34}$. Hence, we have that their intersection is one dimensional. We choose $|e_2\rangle$ a unit vector contained in this intersection, and hence $|e_2\rangle \perp |e_1\rangle$ and $|e_2\rangle \perp |e_3\rangle$. This means that $\{|e_1\rangle, |e_2\rangle, |e_3\rangle\}$ is an orthonormal basis, and $P_1=|e_1\rangle\langle e_1|+|e_2\rangle\langle e_2|$, $P_2=|e_3\rangle\langle e_3|$, $P_3=|e_2\rangle\langle e_2|+|e_3\rangle\langle e_3|$ and $P_4=|e_1\rangle\langle e_1|$.
\qed
\end{proof}

\begin{theorem} \label{th2}
For the data of our experiment exist compatible self-adjoint operators, hence, following the notations introduced, ${\mathcal O}_{12}\not={\mathcal O}_{34}$ such that $[{\mathcal O}_{12}, {\mathcal O}_{34}]=0$, modeling both bets. This means that, again following the notations introduced, that we have
\begin{eqnarray}
&\langle e_1|e_1\rangle=\langle e_2|e_2\rangle=\langle e_3|e_3\rangle=1 \label{start}\\
&\langle e_1|e_2\rangle=\langle e_1|e_3\rangle=\langle e_2|e_3\rangle=0 \\
&|\langle v_{s}^{12}|e_1\rangle|^{2}+|\langle v_{s}^{12}|e_2\rangle|^{2}=0.68 \label{1}\\
&|\langle v_{s}^{12}|e_3\rangle|^{2}=0.32 \label{2} \\
&|\langle v_{s}^{34}|e_2\rangle|^{2}+|\langle v_{s}^{34}|e_3\rangle|^{2}=0.31 \\ \label{e1}
&|\langle v_{s}^{34}|e_1\rangle|^{2}=0.69 \\
&|\langle v_{c}|e_1\rangle|^{2}+|\langle v_{c}|e_2\rangle|^{2}=0.5=|\langle v_{c}|e_3\rangle|^{2} \label{4}\\ \label{vc}
&|\langle v_{c}|e_2\rangle|^{2}+|\langle v_{c}|e_3\rangle|^{2}=0.5=|\langle v_{c}|e_1\rangle|^{2} \label{6}
\end{eqnarray}
\end{theorem}
\begin{proof}
We can explicitly construct an orthonormal basis $\{|e_1\rangle, |e_2\rangle, |e_3\rangle \}$ which simultaneously satisfies Eqs. (\ref{start})--(\ref{6}). We omit the 
explicit construction, for the sake of brevity, and we only report the solution, as follows. The orthonormal vectors $|e_1\rangle$ and $|e_3\rangle$ are respectively given by
\begin{eqnarray}
&&|e_1\rangle=|0.38e^{i61.2^{\circ}},0.13e^{i248.4^{\circ}}, 0.92e^{i194.4^{\circ}}\rangle \\
&&|e_3\rangle=|0.25e^{i251.27^{\circ}},0.55e^{i246.85^{\circ}}, 0.90e^{i218.83^{\circ}}\rangle
\end{eqnarray}
One can then construct at once a unit vector $|e_2\rangle$ orthogonal to both $|e_1\rangle$ and $|e_3\rangle$, 
which we don't do explicitly, again for the sake of brevity.
\qed
\end{proof}
Theorems \ref{th1} and \ref{th2} entail 
that within our quantum modeling approach the two bets can be represented by commuting observables such that the statistical data of the real experiment are faithfully modeled.

Then, the following orthogonal projection operators model the agents' decisions.
\begin{equation}
P_2=|e_3\rangle\langle e_3|=\left( \begin{array}{ccc}
0.06 & 0.14e^{i4.42^\circ} & 0.23e^{i32.44^\circ} \\
0.14e^{-i4.42^\circ} & 0.30 & 0.49e^{i28.02^\circ} \\
0.23e^{-i32.44^\circ} & 0.49e^{-i28.02^\circ} & 0.81 \end{array} \right)
\end{equation}
and
\begin{equation}
P_4=|e_1\rangle\langle e_1|=\left( \begin{array}{ccc}
0.14 & 0.05e^{-i187.2^\circ} & 0.35e^{-i133.2^\circ} \\
0.05e^{i187.2^\circ} & 0.02 & 0.12e^{i54^\circ} \\
0.35e^{i133.2^\circ} & 0.12e^{-i54^\circ} & 0.85 \end{array} \right)
\end{equation}
Thus, $P_1=\mathbbmss{1}-P_2$ and $P_3=\mathbbmss{1}-P_4$ can be easily calculated. 

Let us now come to the representation of the observables. The observable associated with the preference between $f_1$ and $f_2$ is then represented by the self--adjoint operator (spectral decomposition) ${\mathcal O}_{12}$, while the observable associated with the preference between $f_3$ and $f_4$ is represented by the self--adjoint operator (spectral decomposition) ${\mathcal O}_{34}$. More explicitly, if we set $o_1=o_3=1$ and $o_2=o_4=-1$, we get the following explicit representations.
\begin{eqnarray}
{\mathcal O}_{12}&=&P_1-P_2=\mathbbmss{1}-2P_2 \nonumber  \\
&=&\left( \begin{array}{ccc}
0.87 & -0.28e^{i4.42^\circ} & -0.46e^{i32.44^\circ} \\
-0.28e^{-i4.42^\circ} & 0.40 & -0.98e^{i28.02^\circ} \\
-0.45e^{-i32.44^\circ} & -0.98e^{-i28.02^\circ} & -0.62 \end{array} \right)
\end{eqnarray}
\begin{eqnarray}
{\mathcal O}_{34}&=&P_3-P_4=\mathbbmss{1}-2P_4 \nonumber  \\
&=&\left( \begin{array}{ccc}
0.71 & -0.10e^{-i187.2^\circ} & -0.70e^{-i133.2^\circ} \\
-0.10e^{i187.2^\circ} & 0.97 & -0.24e^{i54^\circ} \\
-0.70e^{i133.2^\circ} & -0.24e^{-i54^\circ} & -0.69 \end{array} \right)
\end{eqnarray}
A direct calculation of the commutator operator $[{\mathcal O}_{12},{\mathcal O}_{34}]$ reveals that the corresponding observables are indeed
compatible.

\section{A real vector space analysis\label{realcase}}
We show in this section that the 
possibility of representing the experimental data in Sect. \ref{experiment} by compatible measurements for the bets 
relies crucially 
on our choice of a Hilbert space over complex numbers as a modeling space. Indeed, if a Hilbert space over real numbers is
attempted, no compatible the observables for the bets and the data in Sect. \ref{mod_experiment} can be constructed any longer, 
as our analysis in this section reveals.

Let us 
indeed attempt to represent the Ellsberg paradox situation in the real Hilbert space ${\mathbb R}^{3}$. This comes to allowing only values of $0$ and $\pi$ for the phases of the complex vectors in Sect. \ref{step3}. It is then easy to see that the only conservative superpositions that remain are the extreme states themselves, that is,
\begin{eqnarray}
|v_s^{12}\rangle=|v_{ry}\rangle=|\pm{1 \over \sqrt{3}}, \pm\sqrt{2 \over 3}, 0\rangle \label{vs12real}\\
|v_s^{34}\rangle=|v_{rb}\rangle=|\pm{1 \over \sqrt{3}}, 0, \pm\sqrt{2 \over 3}\rangle \label{vs34real}
\end{eqnarray}
This means that superposition states such as $|v_s^{12}\rangle$ and $|v_s^{34}\rangle$ are only conservative, in case complex numbers are used for the superposition. This is a first instance of the necessity of complex numbers for a quantum modeling of the Ellsberg situation, because indeed, we should be able to represent the priors, hence the quantum states, by superpositions, of the extreme states, and not by the extreme states themselves. But, let us prove that even if we opt for representing the quantum states by the extreme states, that no real Hilbert space representation with compatible observables modeling the bets and are experimental data is possible.

Taking into account the content of Theorems \ref{th1}, and making use of the notations introduced, we can state the following: For a compatible solution to exist, we need to find two unit vectors, let us denote them $|x\rangle$ and $|y\rangle$, such that they are elements of the one dimensional eigenspaces $|x\rangle\in{\mathcal S}_{2}^{12}$ and $|y\rangle\in{\mathcal S}_{4}^{34}$ respectively. For case (i) of Theorems \ref{th1} to be satisfied, this is the case where the self-adjoint operators representing the compatible measurements are different, we have that $|x\rangle$ needs to be orthogonal to $|y\rangle$, which we can express as $\langle x|y\rangle=0$. For the case (ii) of Theorems \ref{th1} to be satisfied, this is the case where the self-adjoint operators representing the compatible measurements are equal, we have that $|x\rangle$ needs to be a multiple of $|y\rangle$, and since both are unit vectors, and we work in a real Hilbert space, this means that $|x\rangle=\pm|y\rangle$. This can be expressed as $\langle x|y\rangle=\pm1$. In the following we will prove that such $|x\rangle$ and $|y\rangle$ do not exist in a real Hilbert space.

In our proof we start by supposing that these vectors exist and find a contradiction. Let us put $|x\rangle=|a,b,c\rangle$ and $|y\rangle=|d,f,g\rangle$. Since $|x\rangle$ and $|y\rangle$ are unit vectors, we have
$1=a^2+b^2+c^2=d^2+f^2+g^2$.
Generalizing Eqs. (\ref{2}), (\ref{e1}), (\ref{4}) and (\ref{vc}) to considering the two cases (i) and (ii) of Theorems \ref{th1} we must have the following equations satisfied for the vectors $|x\rangle$ and $|y\rangle$.
\begin{eqnarray}
{1 \over 2}&=&|\langle {1 \over \sqrt{3}}, {1 \over \sqrt{3}},{1 \over \sqrt{3}}|a,b,c\rangle|^2 \nonumber\\
&=&{1 \over 3}(a+b+c)^2={1 \over 3}(a^2+b^2+c^2+2ab+2ac+2bc) \\
0.69&=&|\langle {1 \over \sqrt{3}}, 0,\sqrt{2 \over 3}|a,b,c\rangle|^2  \nonumber\\
&=&{1 \over 3}(a+\sqrt{2}c)^2={1 \over 3}(a^2+2c^2+2\sqrt{2}ac) \label{equ01} \\
{1 \over 2}&=&|\langle {1 \over \sqrt{3}}, {1 \over \sqrt{3}},{1 \over \sqrt{3}}|d,f,g\rangle|^2 \nonumber\\
&=&{1 \over 3}(d+f+g)^2={1 \over 3}(d^2+f^2+g^2+2df+2dg+2fg) \\
0.32&=&|\langle {1 \over \sqrt{3}}, \sqrt{2 \over 3},0|d,f,g\rangle|^2 \nonumber\\
&=&{1 \over 3}(d+\sqrt{2}f)^2={1 \over 3}(d^2+2f^2+2\sqrt{2}df)
\end{eqnarray}
We stress that we have considered here the $++$ signs choices for $|v_s^{12}\rangle$ and $|v_s^{34}\rangle$. We will later consider the other possibilities. 
Elaborating we get the following set of equations to be satisfied
\begin{eqnarray} \label{unite1}
1&=&a^2+b^2+c^2 \\ \label{conee101}
1.5&=&a^2+b^2+c^2+2ab+2ac+2bc \\ \label{conee102}
2.07&=&a^2+2c^2+2\sqrt{2}ac \label{equ02} \\ \label{unite3}
1&=&d^2+f^2+g^2 \\ \label{conee301}
1.5&=&d^2+f^2+g^2+2df+2dg+2fg \\ \label{conee302}
0.96&=&d^2+2f^2+2\sqrt{2}df
\end{eqnarray}
The points $(a,b,c)$ satisfying Eq. (\ref{conee101}) lie on a cone in the origin and centred around $|v_c\rangle$, and the points $(a,b,c)$ satisfying Eq. (\ref{conee102}) lie on a cone in the origin and centred around $|v_s^{34}\rangle$. Hence Eqs. (\ref{conee101}) and (\ref{conee102}) can only jointly be satisfied where these two cones intersect.  Further need $(a,b,c)$ to be the coordinates of a unit vector, which is expressed by Eq. (\ref{unite1}). For two cones there are in a three dimensional real space only two possibilities, or they cut each other in two lines through the origin, or they do not have an intersection different from the origin. On two lines through the origin, four unit vectors can be found always. This means that Eqs. (\ref{unite1}), (\ref{conee101}) and (\ref{conee102}) have four solutions, or none. We are in a situation here of four solutions, which are the following
\begin{eqnarray}
&&1 \quad a=0.052 \quad b=0.192 \quad c=0.980 \\
&&2 \quad a=-0.931 \quad b=0.065 \quad c=-0.359 \\
&&3 \quad a=-0.052 \quad b=-0.192 \quad c=-0.980 \\ 
&&4 \quad a=0.931 \quad b=-0.065 \quad c=0.359 
\end{eqnarray}
Also the solutions of Eqs. (\ref{unite3}), (\ref{conee301}) and (\ref{conee302}) are four points on the two intersecting lines of a cone, or none, if the cones do not intersect. In the situation corresponding to our experimental data, we also here find four solutions, which are
\begin{eqnarray}
&&1 \quad d=0.969 \quad f=0.008 \quad g=0.248 \\
&&2 \quad d=0.154 \quad f=-0.802 \quad g=-0.577 \\
&&3 \quad d=-0.969 \quad f=-0.008 \quad g=-0.248 \\ 
&&4 \quad d=-0.154 \quad f=0.802 \quad g=0.577 
\end{eqnarray}  
Our proof follows now easily, when we verify that none of these solutions represent allows $|x\rangle$ and $|y\rangle$ to be an orthogonal pair of vectors, or a pair of vectors equal to each other, or to ones opposite. We can verify this by calculating the number $\langle x|y\rangle$, and seeing that they are all different from 0, different from +1, and different from -1.
We indeed have
\begin{eqnarray}
&&\langle e_1|e_3\rangle^{11}=0.296 \quad \langle e_1|e_3\rangle^{12}=-0.711 \\
&& \langle e_1|e_3\rangle^{13}=-0.296 \quad \langle e_1|e_4\rangle^{13}=0.711\\
&&\langle e_1|e_3\rangle^{21}=-0.990 \quad \langle e_1|e_3\rangle^{22}=0.011 \\
&&\langle e_1|e_3\rangle^{23}=0.990 \quad \langle e_1|e_3\rangle^{24}=-0.011 \\
&&\langle e_1|e_3\rangle^{31}=-0.296 \quad \langle e_1|e_3\rangle^{32}=0.711 \\
&& \langle e_1|e_3\rangle^{33}=0.296 \quad \langle e_1|e_4\rangle^{33}=-0.711 \\
&&\langle e_1|e_3\rangle^{41}=0.990 \quad \langle e_1|e_3\rangle^{42}=-0.011 \\
&&\langle e_1|e_3\rangle^{43}=-0.990 \quad \langle e_1|e_3\rangle^{44}=0.011
\end{eqnarray}
To complete our proof of the non-existence of compatible observables in a real Hilbert space, we need to analyse also all the other possibilities, i.e. all possible choices of $+$
and $-$
for $|v_s^{12}\rangle$ and $|v_s^{34}\rangle$. This can be done completely along the same lines of the above, and hence we do not represent it explicitly here. We have however verified all cases carefully, and indeed, none of the possibilities lead to 
vectors 
$|x\rangle$ and $|y\rangle$ such that their in product $\langle x|y\rangle$ equals 0, +1, or -1.
This completes our proof of the impossibility to model our experimental data for the considered bets, such that these bets are represented by commuting self-adjoint operators, hence compatible measurements, in a real Hilbert space.

The above result is relevant, in our opinion, and we think it is worth to discuss it more in detail. The existence of compatible observables to represent the decision-makers' choice among the different acts in our experiment on the Ellsberg paradox is a direct consequence of the fact that we used a complex Hilbert space as a modeling space. As we have proved in this section, if one instead uses a real vector space, then the collected experimental data cannot be reproduced by compatible observables. Hence, one has two possibilities, in this case. Either one requires that compatible observables exist that accord with an Ellsberg-type situation, and then one has to accept a complex Hilbert space representation where ambiguity aversion is coded into superposed quantum states,
and these superpositions are of the `complex type', hence entailing genuine interference -- when superpositions are with complex (non real) coefficient, this means that the quantum effect of interference is present. Alternatively, one can use a representation in a real vector space but, then, one should accept that an Ellserg-type situation cannot be reproduced by compatible observables. In either case, the 
appearance of quantum structures -- interference due to the presence of genuine complex numbers, or incompatibility due to the impossibility to represent the data by compatible measurements -- seems unavoidable in the Ellsberg paradox situation.  

Our quantum theoretic modeling of the Ellsberg paradox situation is thus completed. We however want to add some explanatory remarks concerning the novelty of our approach, as follows.

(i) We have incorporated the subjective preference of traditional economics approaches in the quantum state describing the conceptual Ellsberg entity. At variance with existing proposals, the subjective preference coded in the quantum state can be different for each one of the acts $f_j$, since it is not derived from any prefixed mathematical rule. Therefore, we can naturally explain a situation with $f_1$ preferred to $f_2$ and $f_4$ preferred to $f_3$, without the need of assuming extra hypotheses.

(ii) In the present paper, we have detected genuine quantum aspects in the description of the Ellsberg paradox situation, namely, contextuality, superposition, and the ensuing interference. Further deepening and experimenting most probably will reveal other quantum aspects. Hence the hypothesis that quantum effects, of a conceptual nature, concretely drive decision--makers' behavior in uncertainty situations is warranted.  

(iii) We have focused here on the Ellsberg paradox. But, our quantum theoretic modeling is sufficiently general to cope with various generalizations of the Ellsberg paradox, which are problematical in traditional economics approaches, such as the Machina paradox \cite{ast12}. This opens the way toward the construction of a unified framework extending standard expected utility and modeling `ambiguity laden' situations in economics and decision theory.

\end{document}